%% file: main.tex
\newif\ifmanu  
\newif\ifec    

\manutrue

\ifmanu
    \documentclass[11pt,letterpaper]{article}
    \usepackage[margin=1in]{geometry}
    \usepackage{latexsym,graphicx,amssymb}
    \usepackage{amsmath,enumerate}
    \usepackage{bbm}
    \usepackage{subfigure}
    \usepackage{float}
    \usepackage{epsfig}
    \usepackage{xspace}
    \usepackage{paralist}
    \usepackage{enumerate}
    \usepackage{cases}
    \usepackage{caption}
    \usepackage{multicol}
    \usepackage{graphicx}
    \usepackage{xcolor}
    \usepackage{tikz}
    \usetikzlibrary{decorations.pathreplacing}
    \usepackage{ifthen}
    \usepackage{algorithm}
    \usepackage[noend]{algpseudocode}
    \usepackage{mathtools}
    \usepackage{amsthm}
    \usepackage{amssymb}
    \usepackage{tcolorbox}
    \usepackage{dsfont}

    \newenvironment{proofof}[1]{{\vspace*{5pt} \noindent\bf Proof of #1:  }}{\hfill\rule{2mm}{2mm}\vspace*{5pt}}

    \newtheorem*{observation*}{Observation}
    \newtheorem{theorem}{Theorem}[section]

    \newtheorem{lemma}{Lemma}[section]
    \newtheorem*{lemma*}{Lemma}
    \newtheorem{claim}{Claim}[section]
    \newtheorem*{claim*}{Claim}

    \newtheorem{remark}{Remark}[section]
\else
    \documentclass[format=acmsmall, review=false]{acmart}
    \usepackage{acm-ec-24-proc}
    \usepackage{booktabs} 
    \usepackage[ruled]{algorithm2e} 
    \usepackage{natbib} 
    
    \SetAlFnt{\small}
    \SetAlCapFnt{\small}
    \SetAlCapNameFnt{\small}
    \SetAlCapHSkip{0pt}
    \IncMargin{-\parindent}
    \usepackage{tikz}
    \usepackage{dsfont}
    \usetikzlibrary{decorations.pathreplacing}

    \setcitestyle{authoryear}
    \AtEndPreamble{
    \newtheorem*{observation*}{Observation}
    \newtheorem{claim}[theorem]{Claim}
    
    }
    \newenvironment{proofof}[1]{\begin{proof}[Proof of #1.]}{\end{proof}}



    \acmYear{2024}\copyrightyear{2024}
    \setcopyright{acmlicensed}
    \acmConference[EC '24]{Conference on Economics and Computation}{July 8--11, 2024}{New Haven, CT, USA}
    \acmBooktitle{Conference on Economics and Computation (EC '24), July 8--11, 2024, New Haven, CT, USA}
    \acmDOI{10.1145/3670865.3673487}
    \acmISBN{979-8-4007-0704-9/24/07}
\fi

\input{notations}

\title{Setting Targets is All You Need: \\
Improved Order Competitive Ratio for Online Selection}
\ifmanu
    \date{}
    \author{
    Liyan Chen \thanks{Institute for Interdisciplinary Information Sciences, Tsinghua University. Email: \texttt{chen-ly21@mails.tsinghua.edu.cn}}
    \and
    Nuozhou Sun \thanks{Institute for Interdisciplinary Information Sciences, Tsinghua University. Email: \texttt{snz21@mails.tsinghua.edu.cn}}
    \and
    Zhihao Gavin Tang \thanks{ITCS, Key Laboratory of Interdisciplinary Research of Computation and Economics, Shanghai University of Finance and Economics. \texttt{tang.zhihao@mail.shufe.edu.cn}}
    }
\else
    \author{Liyan Chen}
    \email{chen-ly21@mails.tsinghua.edu.cn}
    \affiliation{
        \institution{Institute for Interdisciplinary Information Sciences, Tsinghua University}
        \country{China}
    }
    \author{Nuozhou Sun}
    \email{snz21@mails.tsinghua.edu.cn}
    \affiliation{
        \institution{Institute for Interdisciplinary Information Sciences, Tsinghua University}
        \country{China}
    }
    \author{Zhihao Gavin Tang}
    \email{tang.zhihao@mail.shufe.edu.cn}
    \orcid{0000-0002-5094-1971}
    \affiliation{
        \department{ITCS, Key Laboratory of Interdisciplinary Research of Computation and Economics}
        \institution{Shanghai University of Finance and Economics}
        \country{China}
    }
    \begin{abstract}
        \input{abstract}

    \end{abstract}
\fi

\begin{document}

\ifmanu
    \maketitle
\else
    
    \maketitle
    
    
\fi

\ifmanu
\begin{abstract}
    \input{abstract}
\end{abstract}
\fi

\section{Introduction}
\input{intro}

\section{Preliminaries}
\label{sec:prelim}
\input{preliminary}

\section{Targeted value algorithms: A Warm-up}
\label{sec:warm-up}
\input{warmup}

\section{Targeted value algorithms with Detection}
\label{sec:improved}
\input{improved}

\section{Hardness Results}
\label{sec:hardness}
\input{hardness}

\ifmanu
\bibliographystyle{plain}
\else
\bibliographystyle{ACM-Reference-Format}
\fi
\bibliography{ref}

\end{document}

%% file: notations.tex
\newcommand{\alg}{\mathsf{ALG}}
\newcommand{\opt}{\mathsf{OPT}}

\newcommand{\tva}{\mathsf{TVA}}
\newcommand{\tvd}{\mathsf{TVD}}
\newcommand{\sta}{\mathsf{STA}}

\newcommand{\Set}[1]{\left\{ #1 \right\}}

\newcommand{\eqdef}{\overset{\mathrm{def}}{=\mathrel{\mkern-3mu}=}}

\newcommand{\dd}[1]{\;\mathrm{d}#1}

\newcommand{\AutoAdjust}[3]{\mathchoice{ \left #1 #2  \right #3}{#1 #2 #3}{#1 #2 #3}{#1 #2 #3} }
\newcommand{\Xcomment}[1]{{}}

\newcommand{\InBrackets}[1]{\AutoAdjust{[}{#1}{]}}
\newcommand{\Ex}[2][]{\operatorname{\mathbb{E}}_{#1}\InBrackets{#2}}
\newcommand{\Exlong}[2][]{\operatornamewithlimits{\mathbb{E}}\limits_{#1}\InBrackets{#2}}
\newcommand{\Prx}[2][]{\operatorname{\mathrm{Pr}}_{#1}\InBrackets{#2}}

\newcommand{\ind}[1]{\mathds{1}\left[\vphantom{\sum}#1\right]}

\newcommand{\argmax}{\mathop{\rm argmax}}

%% file: abstract.tex
There is a rising interest for studying the online benchmark as an alternative of the classical offline benchmark in online stochastic settings. Ezra, Feldman, Gravin, and Tang (SODA 2023) introduced the notion of order-competitive ratio, defined as the worst-case ratio between the performance of the best order-unaware algorithm and the best order-aware algorithm, to quantify the loss incurred by the lack of knowledge of the arrival order. They showed in the online single selection setting (a.k.a. the prophet problem), the optimal order-competitive ratio achieved by deterministic algorithms is $1/\varphi \approx 0.618$, and left with an open question whether randomized algorithms can do better.

We answer the open question firmly by introducing a novel family of algorithms called \emph{targeted value algorithms}.
We show that the task of online selection is as easy as guessing the optimal online benchmark.
Specifically, we provide 1) an alternative optimal $1/\varphi$ order-competitive algorithm by setting the targeted value deterministically, and 2) a $0.732$ order-competitive algorithm by setting the targeted value randomly.
We further provide a $0.758$ upper bound on the order-competitive ratio of our algorithm, showing that our analysis is close to the best possible, and establish an upper bound of $0.829$ on the order-competitive ratio for general randomized order-unaware algorithms.

%% file: intro.tex
Stochastic settings are arguably the most important alternative to the worst-case analysis of online algorithms. Instead of having an almighty adversary controlling the whole instance, stochastic models assume the input to be generated from known or unknown distributions. 

A classical example is the prophet problem, originally studied in the optimal stopping theory~\cite{krengel1977semiamarts,krengel1978}. Consider a gambler facing a sequence of $n$ boxes. Each box is associated with a reward, drawn from independent known distributions. Upon the arrival of each box, the gambler observes its reward and decides whether to select it. 
The decisions are made online and irrevocably, i.e., the gambler cannot retrieve a box once he rejected it. The goal is to maximize the expected selected reward.
It is shown~\cite{krengel1977semiamarts,krengel1978,samuel1984} that the gambler can guarantee at least $\frac{1}{2}$ fraction of the expected largest reward and the constant $\frac{1}{2}$ is the best possible. This result is known as the prophet inequality.

There is a noticeably growing interest for studying prophet inequalities and its variants in the economics and computation community\footnote{According to an incomplete survey, more than 20 papers regarding prophet inequalities are published in the past three years. A tutorial on prophet inequalities is given in EC 2021.}, due to the connection between prophet inequalities and pricing mechanisms~\cite{aaai/HajiaghayiKS07, orl/CorreaFPV19, stoc/ChawlaHMS10}.
Most of these works adapt the standard competitive analysis from online algorithms, that measures the performance of an algorithm by comparing it against the optimal offline optimum (a.k.a., the prophet). This often leads to overly pessimistic evaluations of online algorithms since the offline benchmarks are too strong to compete against.

On the other hand, the notion of \emph{optimal online algorithms} is well-defined for stochastic settings, provided that the distributions and the arrival order are known. For instance, the optimal online algorithm for the classic prophet problem can be computed efficiently via backward induction. Nevertheless, not much work concerns the online optimum benchmark until very recently. 
Quoting a sentence by Karlin and Koutsoupias from Chapter 24 of the textbook ``Beyond the worst-case analysis of algorithms''~\cite{book/Roughgarden20}, the lack of study of the online optimal algorithm is mainly due to its technical difficulty.
\begin{quote}
``Although using the optimal online algorithm as a benchmark makes perfect sense in stochastic settings, it is not so common in the literature, largely because we do not have many techniques for getting a handle on the optimal online algorithm or comparing arbitrary online algorithms in stochastic settings.''
\end{quote}

Ezra et al.~\cite{soda/EzraFGT23} proposed the notion of \emph{order-competitive ratio} to quantify the importance of knowing the arrival order in advance. Specifically, they suggested a paradigm for designing \emph{order-unaware} algorithms (that only knows the distributions in advance but not the arrival order) for online stochastic models, and to compete against the optimal online algorithm. Among other results, they presented an optimal deterministic algorithm for the gambler's problem with an order-competitive ratio of $\frac{1}{\varphi}  \approx 0.618$, surpassing the $0.5$ impossibility result when the benchmark is set to be the prophet.

\subsection{Our Contributions}
In this work, we advance the study of the online benchmark by following the order-competitive analysis by Ezra et al.~\cite{soda/EzraFGT23}. The concept of order-competitive ratio aligns seamlessly with the theme of online algorithms, as it similarly operates within a setting of incomplete information.

We study the stylized online single selection problem (i.e. the gambler's game), where we intentionally avoid using the term ``prophet inequality'' for referring to the setting, since we are interested in comparing against the online benchmark.
Our main result is a positive answer to the major open question left by Ezra et al.~\cite{soda/EzraFGT23}, confirming that randomized algorithms can achieve better order-competitive ratios than deterministic algorithms. Though not surprising, this is in contrast to the case for competing against the offline optimum. Recall that deterministic algorithms and randomized algorithms have the same worst-case competitive ratio of $\frac{1}{2}$.
\begin{theorem}
    There exists a (randomized) $~0.732$ order-competitive algorithm for the online single selection problem.
\end{theorem}

\paragraph{Our Techniques.}
Our algorithm and analysis are built on the following structural property of the online selection problem that might be of independent interests. To the best of our knowledge, this property is missed in the previous rich literature of prophet inequalities.
\begin{observation*}
As long as the expected payoff of the optimal online algorithm is known in advance, there exists a deterministic order-unaware algorithm achieving the same expected payoff.
\end{observation*}
In other words, the information of the arrival order can be compressed into a single number, and the optimal online algorithm can then be learned on the fly with this number given as an advice. 
To this end, we introduce a novel family of algorithms that we name as \emph{targeted value algorithms}. Intuitively, we make a guess of the optimal expected payoff and set it as the targeted value. An ideal algorithm should have the same performance of the optimal online algorithm when the targeted value is set accurately, and its performance degrades smoothly with respect to the error of the guess.
Indeed, we manage to achieve at least an expected reward of the targeted value, as long as the targeted value is an underestimation of the optimal payoff.

On the other hand, an algorithm unavoidably losses a significant amount of payoff when the targeted value becomes an overestimation of the optimal payoff. 
Designing algorithms with a robustness guarantee when the targeted value is an overestimation turns out to be the most challenging and technical part of our result.
We provide a warm-up version of the targeted value algorithm in Section~\ref{sec:warm-up} which conveys most of the important ideas of our work, and a more technically involved algorithm in Section~\ref{sec:improved} with a better robustness guarantee.
The most crucial lemmas are Lemma~\ref{lem:targeted_value} and ~\ref{lem:targeted_value_detection}, which give a formal proof of the above observation and establish the robustness guarantees. Finally, we optimize a randomized strategy for setting the targeted value to achieve the stated order-competitive ratio.

As a byproduct of our approach, we prove that when the targeted value is set to be $\frac{1}{\varphi}$ of the prophet, our algorithm is $\frac{1}{\varphi}$ order-competitive against the online optimum. This gives an alternative optimal deterministic order-competitive algorithm for the online single selection problem. We believe that our algorithm and analysis are simpler and more intuitive than that of Ezra et al.~\cite{soda/EzraFGT23}. 

\paragraph{Hardness Results.}
We complement our algorithmic results with two hardness results. The first hardness result establishes an upper bound $0.829$ of the order-competitive ratio for any randomized order-unaware algorithms. 
The second hardness result establishes an upper bound $0.758$ of the order-competitive ratio for our (randomized) targeted value algorithms, showing that our analysis is close to the best possible. 

\subsection{Related Works}
Ezra and Garbuz~\cite{wine/EzraG23} examined the notion of order-competitive ratio in various combinatorial settings, including multi-unit selections, downward-closed feasibility constraints, and general (non downward-closed) feasibility constraints. Their results are mostly negative.

Another line of works concerns the computational complexity of the online optimal benchmark. 
Sepcifically, Papadimitriou et al.~\cite{ec/PapadimitriouPS21} studied the online stochastic matching problem with known arrival order. They proved that the optimal algorithm is PSPACE-hard to approximate within some constant $1-\omega(1)$, and designed a $0.51$-approximation algorithm\footnote{Notice that this is a full information setting. Hence, we use the terminology approximation ratio instead of competitive ratio.}. The ratio is improved by a series of work~\cite{icalp/SaberiW21, ec/BravermanDL22,corr/NaorSW23} and the current best bound is $0.652$ by Naor et al.~\cite{corr/NaorSW23}.

We remark that finding the optimal order-competitive algorithm corresponds to computing the optimal algorithm with respect to a known stochastic order. Indeed, consider a two player zero-sum game where one player plays (randomized) algorithms and the other player plays (randomized) orders. By the minimax theorem, we have 
\[
\max_{\text{alg}\sim A} \min_{\text{order}\sim O} \Ex{\text{Reward}(\text{alg}, \text{order})} = \min_{\text{order}\sim O} \max_{\text{alg} \sim A} \Ex{\text{Reward}(\text{alg}, \text{order})}.
\]
To the best of our knowledge, the only work concerning the computation complexity of the optimal online algorithm with respect to a stochastic order is by D{\"{u}}tting et al.~\cite{ec/DuttingGRTT23}, in which a PTAS is provided for the prophet secretary problem, i.e., when the arrival order is drawn uniformly at random from all possible orders.

%% file: preliminary.tex
Consider an online single-selection problem with $n$ boxes. Each box $t$ is associated with a value $v_t$, that is drawn independently from distribution $F_t$. 
The distributions of the boxes are known in advance while the values of boxes are revealed in a sequence. The algorithm needs to decide whether to select the box immediately and irrevocably; and the goal is to maximize the expectation of the selected value.

\paragraph{Online Benchmark / Order-aware algorithms.}
An \emph{order-aware} online algorithm knows the arrival order of all boxes in advance. Throughout the paper, we use $\opt$ to denote the optimal order-aware algorithm and its expected reward. By a standard backward induction analysis,
\[
\opt_t = \Exlong[v_t]{\max \left( v_t, \opt_{t+1} \right)}, \quad \forall t \in [n], \quad \text{ where } \opt_{n+1} = 0.
\]
Here, $\opt_t$ denotes the expected reward of the optimal order-aware algorithm from box $t$ till the end, and $\opt = \opt_1$. This is also known as the optimal online benchmark. Recall that the optimal offline benchmark (a.k.a. prophet) is defined as $\Ex{\max_t v_t}$.

\paragraph{Order-Competitive Ratio / Order-unaware algorithms.}
We focus on the order-competitive ratio that is introduced recently by Ezra et al.~\cite{soda/EzraFGT23}. 
Formally, an algorithm is called \emph{order-unaware} if it only knows the set of distributions $\{F_t\}$ in advance and then learns the identity $t$ of the $t$-th box upon its arrival. In other words, the algorithm has no priori information about the arrival order. The order-competitive ratio of an order-unaware algorithm is defined as $\min_{\pi} \frac{\alg(\pi)}{\opt(\pi)}$, where $\pi$ denotes the arrival order of an instance. This quantity measures the value of knowing the arrival order. Recall that the classical notion of competitive ratio is defined as the ratio between the expected reward of the algorithm and the offline benchmark.

\subsection{Single-threshold algorithms.}
An important family of order-unaware algorithms are called single-threshold algorithms that are parameterized by a fixed threshold $\tau$. The algorithm then selects the first box whose value is at least $\tau$. 
This family of algorithms is known to achieve the optimal $\frac{1}{2}$ competitive ratio against the offline benchmark, but is unable to achieve a better than $\frac{1}{2}$ order-competitive ratio against the online benchmark. 

Nevertheless, we shall consider single-threshold algorithms as an intermediary for lower bounding the performance of our order-unaware algorithms. Specifically, we need the following lower bound on the performance of the single-threshold algorithm that is folklore to the prophet inequality community. For completeness, we provide a proof here. 
Here, we define $x^+ \eqdef \max (x, 0)$ for every $x \in \mathbb{R}$.

\begin{lemma}
\label{lem:single-threshold}
Let $\sta(\tau)$ be the expected reward of a single-threshold algorithm with threshold $\tau$. Then we have the following.
\begin{equation*}
\sta(\tau) \ge \Prx{\max_{i} v_i \geq \tau} \cdot \tau + \Prx{\max_i v_i < \tau} \Ex{\left(\max_i v_i - \tau \right)^+}.
\end{equation*}
\end{lemma}
\begin{proof}
The two terms on the right hand side are known as the ``revenue'' and ``utility'' of the algorithm in the literature and we follow the standard revenue-utility analysis. 
\begin{align*}
\sta(\tau) & = \sum_{t} \Prx{\max_{i<t} v_i < \tau} \cdot \Ex{v_t \cdot \ind{v_t \ge \tau}} \\
& = \sum_{t} \Prx{\max_{i<t} v_i < \tau} \cdot \left( \Prx{v_t \ge \tau} \cdot \tau + \Ex{ \left( v_t - \tau \right)^+} \right) \\
& = \Prx{\max_{i} v_i \geq \tau} \cdot \tau + \sum_{t} \Prx{\max_{i<t} v_i < \tau} \cdot \Ex{ \left( v_t - \tau \right)^+} \\
& \ge \Prx{\max_{i} v_i \geq \tau} \cdot \tau + \sum_{t} \Prx{\max_{i} v_i < \tau} \cdot \Ex{ \left( v_t - \tau \right)^+} \\
& = \Prx{\max_{i} v_i \geq \tau} \cdot \tau + \Prx{\max_{i} v_i < \tau} \cdot   \Ex{ \sum_{t} \left( v_t - \tau \right)^+} \\
& \ge \Prx{\max_{i} v_i \geq \tau} \cdot \tau + \Prx{\max_i v_i < \tau} \Ex{\left(\max_i v_i - \tau \right)^+}.
\end{align*}
The first equation considers the expected gain from each box; the last inequality follows from the fact that $\sum_x x^+ \ge \max_x x^+$.
\end{proof}

%% file: warmup.tex
For the convenience of notations, we assume the boxes arrive in a specific order from $1, 2, \ldots, n$, and we will observe the identity of the $t$-th box and its realized value $v_t \sim \mathcal F_t$ at stage $t$. The order-unawareness of our algorithm shall be clear from its description. 

\paragraph{Targeted value algorithms (TVA)} 
Our algorithm is parameterized by a \emph{targeted value} $g_0$, meaning that we aim to collect an expected reward of $g_0$ from the instance.

At each stage $t \in [n]$, our algorithm first updates our targeted value $g_t$ for the future, based on the identity of the $t$-th box and our estimation $g_{t-1}$ from the last step. Specifically, we define
\[
g_t \eqdef \min \left\{ x \ge 0 \left| \Exlong[v_t]{\max (v_t, x)} \ge g_{t-1}  \right. \right\}.
\]
Notice that such a solution always exists since $\Ex[v_t]{\max (v_t, x)}$ is a continuously non-decreasing function of $x$. Moreover, $g_t \le g_{t-1}$ since  $\Ex[v_t]{\max (v_t, g_{t-1})} \ge g_{t-1}$. Then we accept the $t$-th box if and only if $v_t \ge g_t$.

\paragraph{Analysis.} Our analysis involves two cases depending on $g_0$ being an underestimation or an overestimation of $\opt$, in a similar flavor of the consistency-robustness analysis in the literature of algorithm design with predictions.
Specifically, our algorithm is \emph{consistent} when $g_0$ is an underestimation of $\opt$ and it collects at least an expected reward of $g_0$; our algorithm is \emph{robust} when $g_0$ is an overestimation of $\opt$ and it collects at least an expected reward of $\Ex{\max_{t} v_t} - g_0$.
\begin{lemma}
\label{lem:targeted_value}
For an arbitrary arrival order, we have the following:
\begin{itemize}
\item if $g_0 \le \opt$, then $\tva(g_0) \ge g_0$;
\item if $g_0 > \opt$, then $\tva(g_0) \ge \Ex{\max_{t} v_t} - g_0$.
\end{itemize}
\end{lemma}

We first discuss implications of this lemma and provide the proof at the end of this section. 
\begin{theorem}
Let $g_0 = \frac{1}{\varphi} \cdot \Ex{\max_t v_t}$, where $\varphi \approx 1.618$ is the golden ratio. Our algorithm with a targeted value of $g_0$ has an order-competitive ratio of $\frac{1}{\varphi} \approx 0.618$.
\end{theorem}
\begin{proof}
For an arbitrary arrival order, there are 2 cases:
\begin{itemize}
    \item If $g_0\le \opt$, then $\tva(g_0) \ge g_0=\frac{1}{\varphi}\cdot \Ex{\max_t v_t}\ge \frac{1}{\varphi}\opt$.
    \item If $g_0\ge \opt$, then $\tva(g_0) \ge \Ex{\max_t v_t}-g_0=\left(\varphi-1\right)\cdot g_0= \frac{1}{\varphi}g_0\ge \frac{1}{\varphi}\opt$.
\end{itemize}
    The first inequality in both cases is due to Lemma \ref{lem:targeted_value}. Thus, the algorithm achieves an order-competitive ratio of $\frac{1}{\varphi}$.
\end{proof}
This provides an alternative optimal deterministic $\frac{1}{\varphi}$ order-competitive algorithm. We argue that it is simpler and more intuitive than the algorithm of Ezra et al.~\cite{soda/EzraFGT23}. Furthermore, by applying a random choice of the targeted value, we confirm that randomized algorithms beat deterministic algorithms for competing against the optimal online benchmark. In contrast, in the classical setting for competing against the optimal offline benchmark, it is known that randomized algorithms cannot achieve a better competitive ratio than deterministic algorithms in the worst case.

\begin{theorem}
    There exists a $0.656$ order-competitive randomized algorithm.
\end{theorem}
\begin{proof}
Let 
$\rho(x) = 
\begin{cases}
0, & x \in [1/2, c) \\
\frac{\Gamma}{2x-1}, & x \in [c,1] 
\end{cases}$, where $c \approx 0.523$ is the solution to $\ln\left(\frac{1}{2 c - 1}\right) - 2 c = 2$, and $\Gamma = \frac{2}{\ln(1/(2c-1))} \approx 0.656$. It is straightforward to verify that $\rho$ is a valid probability density function, i.e. $\int_{1/2}^1 \rho(x) \dd x = 1$.

Consider the targeted value algorithm, with a randomized targeted value of $g_0 = x \cdot \Ex{\max_{i} v_i}$, where $x$ is sampled from $\rho(x)$. We claim that this randomized algorithm has an order-competitive ratio of at least $\Gamma$. 
Indeed, for an arbitrary arrival order, let $y =\opt / \Ex{\max_{i} v_i}$ that lies in between $\frac{1}{2}$ and $1$. According to Lemma~\ref{lem:targeted_value}, we have
\[
\Ex{\tva(g_0)} \ge \int_{\frac{1}{2}}^{y} x \cdot \Ex{\max v_i} \rho(x) \dd x + \int_{y}^1 (1-x) \cdot \Ex{\max v_i} \rho(x) \dd x,
\]
where the first integration corresponds to the underestimation case and the second integration corresponds to the overestimation case.
If $y < c$, the right hand side equals
\[
\Ex{\max v_i} \cdot \int_{c}^{1} (1-x) \cdot \frac{\Gamma}{2x-1} \dd x = \Gamma \cdot c \cdot \Ex{\max v_i} \ge \Gamma \cdot \opt~.
\]
Else $y \ge c$, the right hand side equals
\[
\Ex{\max v_i} \cdot \left( \int_{c}^{y} x \cdot \frac{\Gamma}{2x-1} \dd x  + \int_{y}^1 (1-x) \cdot \frac{\Gamma}{2x-1} \dd x \right) = \Gamma \cdot y \cdot \Ex{\max v_i} = \Gamma \cdot \opt~.
\]
This concludes the proof of the theorem.
\end{proof}

\subsection{Proof of Lemma~\ref{lem:targeted_value}}
Throughout the proof, we fix an arbitrary arrival order and use $\{g_t\}_{t=0}^{n}$ to denote the targeted values computed by our algorithm. We use $\tva_t$ (resp. $\opt_t$) to denote the expected reward of our algorithm (resp. the optimal algorithm) from stage $t$ till the end. Then we have
\begin{align*}
& \tva_t = \Exlong[v_t]{v_t \cdot \ind{v_t \ge g_t}} + \Prx{v_t < g_t} \cdot \tva_{t+1} & \forall t \in [n], \text{ where } \tva_{n+1} = 0; \\
& \opt_t = \Exlong[v_t]{\max \left( v_t, \opt_{t+1} \right)} & \forall t \in [n], \text{ where } \opt_{n+1} = 0.
\end{align*}

We start with the following observation that is crucial to our analysis: if $g_0$ is an underestimation (resp. overestimation) of the optimal value, our targeted value $g_t$ at each step stays as an underestimation (resp. overestimation) of the future. Formally, we prove the following.
\begin{claim}
\label{cl:underestimate}
If $g_0 \le \opt_1$, then $g_t \le \opt_{t+1}$ for all $t \in [n]$; else, $g_t \ge \opt_{t+1}$ for all $t \in [n]$.
\end{claim}
\begin{proof}
We only prove the statement for the underestimation case with $g_0 \le \opt_1$ by induction. The analysis applies almost verbatim to the overestimation case. The base case when $t=0$ is the assumption of the statement. Next, suppose $g_{t-1} \le \opt_{t}$. 
Then we have
\[
\Exlong[v_t]{\max \left( v_t, \opt_{t+1} \right)} = \opt_t \ge g_{t-1}, \text{ while } g_t = \min \left\{ x \ge 0 \left| \Exlong[v_t]{\max (v_t, x)} \ge g_{t-1}  \right. \right\}.
\]
Therefore, $\opt_{t+1} \ge g_t$ and we conclude the proof.
\end{proof}

\paragraph{Underestimation ($g_0 \le \opt$).} 
Now, we are ready to prove the first statement of the lemma (i.e., when $g_0 \le \opt_1$). We prove that $\tva_t \ge g_{t-1}$ for all $t \le n+1$ by backward induction. The base case when $t=n+1$ is trivial since $\tva_{n+1}=0$ and $g_n \le \opt_{n+1}=0$ by the above claim. Next, suppose $\tva_{t+1} \ge g_t$. We then have
\begin{align*}
\tva_t & = \Exlong[v_t]{v_t \cdot \ind{v_t \ge g_t}} + \Prx{v_t < g_t} \cdot \tva_{t+1} \\
& \ge \Exlong[v_t]{v_t \cdot \ind{v_t \ge g_t}} + \Prx{v_t < g_t} \cdot g_t = \Exlong[v_t]{\max(v_t, g_t)} \ge g_{t-1},
\end{align*}
where the first inequality follows by the induction hypothesis and the last inequality follows by the definition of $g_t$.
\paragraph{Overestimation ($g_0 \ge \opt$).}
To prove the second statement of the lemma, we consider the single-threshold algorithm with threshold $g_0$ for analysis purpose. We use $\sta_t$ to denote the expected reward of this single threshold algorithm from stage $t$ till the end. That is,
\[
\sta_t = \Exlong[v_t]{v_t \cdot \ind{v_t \ge g_0}} + \Prx{v_t < g_0} \cdot \sta_{t+1}, \quad \forall t \in [n], \text{ where } \sta_{n+1} = 0.
\]
\begin{claim}
    For all $t \le n+1$, we have $\tva_t \ge \sta_t$.
\end{claim}
\begin{proof}
We prove that statement by backward induction on $t$. The base case when $t=n+1$ is trivial. Suppose $\tva_{t+1} \ge \sta_{t+1}$. Then we have,
\begin{align*}
\tva_t & = \Exlong[v_t]{v_t \cdot \ind{v_t \ge g_t}} + \Prx{v_t < g_t} \cdot \tva_{t+1} \\
& \ge \Exlong[v_t]{v_t \cdot \ind{v_t \ge g_0}} + \Prx{v_t < g_0} \cdot \tva_{t+1} \\
& \ge \Exlong[v_t]{v_t \cdot \ind{v_t \ge g_0}} + \Prx{v_t < g_0} \cdot \sta_{t+1} = \sta_t~.
\end{align*}
Here, the first inequality follows from the fact that $\tva_{t+1} \le \opt_{t+1} \le g_t \le g_0$; the second inequality follows from the induction hypothesis. 
\end{proof}
Finally, by Lemma~\ref{lem:single-threshold}, we concludes the proof of the lemma:
\begin{align*}
\tva_1 \ge \sta_1 & \ge \Prx{\max_t v_t \ge g_0} \cdot g_0 + \Prx{\max_{t} v_t < g_0} \cdot \Exlong{(\max_{t} v_t - g_0)^+} \\
& \ge \min \left( g_0, \Exlong{(\max_{t} v_t - g_0)^+} \right) \ge \Exlong{\max_{t} v_t} - g_0~,
\end{align*}
where the last inequality follows from the fact that $g_0 \ge \opt \ge \frac{1}{2} \Ex{\max_{t} v_t}$.



%% file: improved.tex
The targeted value algorithms that we present in Section~\ref{sec:warm-up} have already conveyed the most important ideas of our approach.
Indeed, it performs the same as the optimal order-aware algorithm for any arrival order, as long as our targeted value $g_0$ is set correctly; and its performance degrades smoothly when $g_0$ is an underestimation. The only defect is its performance when $g_0$ is an overestimation. 
In this section, we provide a modification of the targeted value algorithms and aims to improve its performance when the targeted value is an overestimation. This would then lead to a randomized algorithm with an improved order-competitive ratio.
Our main result is the following.

\begin{theorem}
\label{thm:main_rand}
    There exists a $0.732$ order-competitive randomized algorithm.
\end{theorem}

Recall Claim~\ref{cl:underestimate} states that if $g_0$ is an overestimation of the optimal payoff, the estimation $g_t$ at each step would remain as an overestimation of the future payoff. However, we are able to figure out the overestimation after certain stage, since at least at the last moment when the algorithm reaches the last stage, our estimation $g_n$ is strictly greater than $0$ which is obviously sub-optimal.

To this end, we introduce a two-stage algorithm which we name as \emph{targeted value algorithms with detection}, that switches to a more conservative mode when the algorithm is aware of its overly aggressive behavior. Formally, the algorithm works as the following.

\paragraph{Targeted Value algorithms with Detection (TVD)}
As before, our modified algorithm sets up a targeted value $g_0$. At each stage $t \in [n]$, we first update $g_t = \min \left\{ x \ge 0 \left| \Exlong[v_t]{\max (v_t, x)} \ge g_{t-1}  \right. \right\}$ and then do a detection by comparing $g_t$ and $\Ex{\max_{i \ge t+1} v_i}$.
\begin{itemize}
    \item If $g_t \le \Ex{\max_{i > t} v_i}$, we accept the $t$-th box if and only if $v_t \ge g_t$;
    \item If $g_t > \Ex{\max_{i > t} v_i}$, we switch to the conservative mode by using a single-threshold $\tau_t$ from the $t$-th box till the end, where
    \[
    \tau_t = \argmax_{\tau} \left( \Prx{\max_{i \ge t} v_i \ge \tau} \cdot \tau + \Prx{\max_{i \ge t} v_i < \tau} \cdot \Ex{\left(\max_{i \ge t} v_i -\tau\right)^+} \right)~.
    \]
\end{itemize}
We use $\tvd(g_0)$ to denote the expected reward of our algorithm. The next lemma states a better performance guarantee for certain ranges of $g_0$ compared to Lemma~\ref{lem:targeted_value}.
\begin{lemma}
\label{lem:targeted_value_detection}
For an arbitrary arrival order, the targeted value algorithm with detection satisfies the following:
\begin{itemize}
\item if $g_0 \le \mathsf{OPT}$, then $\tvd(g_0) \ge g_0$;
\item if $g_0 > \mathsf{OPT}$, then $\tvd(g_0) \ge \max\{ \Ex{\max_{t} v_t} - g_0, \frac{1}{2} g_0 \}$.
\end{itemize}
\end{lemma}
\begin{proof}
The analysis for the underestimation case is identical to the proof of Lemma~\ref{lem:targeted_value}, since we would not detect an overestimation throughout the instance and the modified algorithm has exactly the same behavior as the original targeted value algorithm.

Next, we focus on the overestimation case. Let $s \in [n]$ be the stage when we detect $g_s > \Ex{\max_{i \ge s+1} v_i}$ and switch to the single-threshold mode. 

We first establish the new lower bound of $\frac{1}{2} g_0$. Indeed, we prove the following stronger statement.
\begin{claim}
    For every $t \le s$, $\tvd_t \ge g_{t-1} - \frac{1}{2} g_{s-1}$.
\end{claim}
\begin{proof}
We prove the statement by backward induction. The base case when $t=s$ holds according to the following:
\begin{align*}
\tvd_s \ge & \Prx{\max_{i \ge s} v_i \ge \tau_s} \cdot \tau_s + \Prx{\max_{i \ge s} v_i < \tau_s} \cdot \Ex{\left(\max_{i \ge s} v_i -\tau_s \right)^+} \\
\ge & \Prx{\max_{i \ge s} v_i \ge \tau} \cdot \tau + \Prx{\max_{i \ge s} v_i < \tau} \cdot \Ex{\left(\max_{i \ge s} v_i -\tau \right)^+} \\
\ge & \Prx{\max_{i \ge s} v_i \ge \tau} \cdot \tau + \Prx{\max_{i \ge s} v_i < \tau} \cdot \left( \Ex{\max_{i \ge s} v_i} -\tau \right) = \frac{1}{2} \Ex{\max_{i \ge s} v_i} \ge \frac{1}{2} g_{s-1},
\end{align*}
where $\tau = \frac{1}{2} \Ex{\max_{i \ge s} v_i}$. The first inequality holds by Lemma~\ref{lem:single-threshold}; the second inequality holds according to the definition of $\tau_s$; and the last inequality holds since stage $s$ is the first time that $g_t > \Ex{\max_{i > t} v_i}$.
Now, suppose $\tvd_{t+1} \ge g_{t} - \frac 1 2 g_{s-1}$. Then we have
\begin{align*}
\tvd_{t} & = \Exlong[v_t]{v_t \cdot \ind{v_t \ge g_t}} + \Prx{v_t < g_t} \cdot \tvd_{t+1} \\
& \ge \Exlong[v_t]{v_t \cdot \ind{v_t \ge g_t}} + \Prx{v_t < g_t} \cdot \left(g_t - \frac{1}{2} g_{s-1} \right) \\
& \ge \Exlong[v_t]{\max (v_t, g_t)} - \frac{1}{2} g_{s-1} \ge g_{t-1} - \frac{1}{2} g_{s-1}.
\end{align*}
This concludes the proof of the claim.
\end{proof}
As an immediate implication, we have $\tvd_1 \ge g_0 - \frac{1}{2} g_{s-1} \ge \frac{1}{2} g_0$.

Finally, we prove that $\tvd(g_0) \ge \Ex{\max_{t} v_t} - g_0$. The analysis is similar to the proof of Lemma~\ref{lem:targeted_value}. Here, we use a two-threshold algorithm as an intermediary step to lower bound the performance of our modified algorithm.
Consider an algorithm $\overline{\alg}$ that uses a fixed threshold $g_0$ for the first $s-1$ stages and uses a fixed threshold $\tau_s$ from the $s$-th stage till the end. We use $\overline{\alg}_t$ to denote the expected reward of $\overline{\alg}$ from stage $t$ till the end.
\begin{claim}
    For every $t \in [n]$, $\tvd_t \ge \overline{\alg}_t$.
\end{claim}
\begin{proof}
The statement holds trivially for $t \ge s$, since the two algorithms have the same behavior. For $t \le s$, we prove the statement by backward induction. Suppose $\tvd_{t+1} \ge \overline{\alg}_{t+1}$. Then we have
\begin{align*}
\tvd_t & = \Exlong[v_t]{v_t \cdot \ind{v_t \ge g_t}} + \Prx{v_t < g_t} \cdot \tvd_{t+1} \\
& \ge \Exlong[v_t]{v_t \cdot \ind{v_t \ge g_0}} + \Prx{v_t < g_0} \cdot \tvd_{t+1} \\
& \ge \Exlong[v_t]{v_t \cdot \ind{v_t \ge g_0}} + \Prx{v_t < g_0} \cdot \overline{\alg}_{t+1} = \overline{\alg}_t.
\end{align*}
Here, the first inequality follows from the fact that $g_0 \ge g_t$ and $\tvd_{t+1} \le \opt_{t+1} \le g_t$; the second inequality follows from the induction hypothesis.
\end{proof}
Finally, we bound the performance of $\overline{\alg}_1$.
\begin{align*}
\overline{\alg}_1 \ge & \Prx{\max_{i < s} v_i \geq g_0} \cdot g_0 + \Prx{\max_{i < s} v_i < g_0} \cdot \Ex{\left(\max_{i < s} v_i - g_0 \right)^+} + \Prx{\max_{i < s} v_i < g_0} \cdot \overline{\alg}_s \\
\geq & \Prx{\max_{i < s} v_i \geq g_0} \cdot g_0 + \Prx{\max_{i < s} v_i < g_0} \cdot \Ex{\left(\max_{i < s} v_i - g_0 \right)^+} \\
& + \Prx{\max_{i < s} v_i < g_0} \cdot \left( \Prx{\max_{i \ge s} v_i \ge \tau_s} \cdot \tau_s + \Prx{\max_{i \ge s} v_i < \tau_s} \cdot \Ex{\left(\max_{i \ge s} v_i -\tau_s \right)^+} \right) \\
\geq & \Prx{\max_{i < s} v_i \geq g_0} \cdot g_0 + \Prx{\max_{i < s} v_i < g_0} \cdot \Ex{\left(\max_{i < s} v_i - g_0 \right)^+} \\
& + \Prx{\max_{i < s} v_i < g_0} \cdot \left( \Prx{\max_{i \ge s} v_i \ge g_0} \cdot g_0 + \Prx{\max_{i \ge s} v_i < g_0} \cdot \Ex{\left(\max_{i \ge s} v_i - g_0 \right)^+} \right) \\
\ge & \Prx{\max_{i} v_i \geq g_0} \cdot g_0 + \Prx{\max_{i} v_i < g_0} \cdot \Ex{\left(\max_{i} v_i - g_0\right)^+} \\
\ge & \min \left( g_0, \Ex{\max_i v_i} - g_0 \right) = \Ex{\max_{i} v_i} - g_0 .
\end{align*}
Here, the first inequality follows by separating the gain of $\overline{\alg}$ into two parts, i.e., before stage $s$ and after stage $s$ and Lemma~\ref{lem:single-threshold}. The second and third inequality follows from Lemma~\ref{lem:single-threshold} and the definition of $\tau_s$. The fourth inequality holds by that $\Prx{\max_{i<s} v_i < g_0} \ge \Prx{\max_{i} v_i < g_0}$ and $\Ex{a^+} + \Ex{b^+} \ge \Ex{(\max(a,b))^+}$. The last equation holds by the fact that $g_0 \ge \opt \ge \frac{1}{2} \Ex{\max_{i} v_i}$. Together with the above claim, we conclude the proof of the lemma.
\end{proof}

\subsection{Proof of Theorem~\ref{thm:main_rand}}
Equipped with Lemma~\ref{lem:targeted_value_detection}, it suffices to design a proper distribution of the targeted value $g_0$ so that our algorithm competes against the online optimal algorithm for an arbitrary arrival order.
More specifically, we design a distribution with probability density function $\rho(x)$ with the following condition so that $\Gamma$ is as large as possible.
\begin{equation}
\label{eqn:main_ratio}
\int_{\frac{1}{2}}^{y} x \cdot \rho(x) \dd x +\int_{y}^1 \max\left( \frac{1}{2}x,1-x \right) \cdot \rho(x) \dd x \ge \Gamma \cdot y, \qquad \forall y \in \left[ \frac{1}{2},1 \right].
\end{equation}
Then we run the targeted value algorithm with detection, with a randomized targeted value of $g_0 = x \cdot \Ex{\max_{i} v_i}$, where $x$ is sampled from $\rho(x)$. We claim that this randomized algorithm has an order-competitive ratio of at least $\Gamma$. Indeed,
For an arbitrary arrival order, let $y =\opt / \Ex{\max_{i} v_i}$ that lies in between $\frac{1}{2}$ and $1$. According to Lemma~\ref{lem:targeted_value_detection}, we have
\begin{multline*}
\Ex{\tvd(g_0)} \ge \int_{\frac{1}{2}}^{y} x \cdot \Ex{\max v_i} \cdot \rho(x) \dd x +\int_{y}^1 \max\left( \frac{1}{2}x,1-x \right) \cdot \Ex{\max v_i} \cdot \rho(x) \dd x \\
\ge \Gamma \cdot y \cdot \Ex{\max v_i} = \Gamma \cdot \opt,
\end{multline*}
where the first integration corresponds to the underestimation case and the second integration corresponds to the overestimation case.

To conclude the proof of the theorem, we provide the explicit construction of the distribution $\rho(x)$ and the corresponding ratio $\Gamma$. We omit the tedious calculation for verifying the stated condition~\eqref{eqn:main_ratio}. Let
$ \rho(x) = \begin{cases}
        \frac{\Gamma}{2x-1}, & c < x \leq \frac 2 3 \\
        \frac{2\Gamma}{x}, & \frac 2 3 < x \leq 1
    \end{cases}$
be the probability density function of $D(x)$, where $c \approx 0.555$ is the unique solution to $\frac{1}{6 c - 3} = e^{2 c}$, and $\Gamma = - \frac{2}{\ln\left( \frac{16}{27} (2 c - 1)\right)} \approx 0.732$ is the order-competitive ratio of our algorithm.

%% file: hardness.tex
We complement our algorithmic results with two hardness results. The first hardness result establishes an upper bound of the order-competitive ratio for any randomized order-unaware algorithms. 
The second hardness result focuses on the targeted value algorithms with detection, showing that our analysis is close to the best possible. 

\begin{theorem}
\label{thm:hardness_general}
    No randomized order-unaware algorithm achieves an order-competitive ratio strictly better than $\Gamma \eqdef \frac{2\mathrm{e}^{\frac{\sqrt5}{2}}}{3\sqrt{\mathrm{e}}+2\mathrm{e}^{\frac{\sqrt5}{2}}-\sqrt{5\mathrm{e}}}\approx 0.8293$.
\end{theorem}

\begin{theorem}
\label{thm:hardness_tvd}
    The (randomized) targeted value algorithms with detection are at most $0.7582$ order-competitive.
\end{theorem}

\subsection{Proof of Theorem~\ref{thm:hardness_general}}
\input{worst-case}

\subsection{Proof of Theorem~\ref{thm:hardness_tvd}}
\input{hardness_tvd}

%% file: worst-case.tex
We revisit the hard instance constructed by Ezra et al.~\cite{soda/EzraFGT23} and study how randomized order-unaware algorithms can improve upon deterministic ones.
Recall the instance of Ezra et al.~\cite{soda/EzraFGT23}:
\begin{itemize}
    \item There exists a ``free reward'' box, whose value equals $\frac{1}{\delta}$ with probability $\delta$, and $0$ otherwise.
    \item There are a set of deterministic boxes with values $\varphi, \varphi - \epsilon, \varphi - 2 \epsilon, \cdots, 1$.
\end{itemize}

Fix an arbitrary order-unaware randomized algorithm and consider the order $\pi$ where the deterministic boxes arrive in a decreasing order with respect to their values, and the free reward box arrives at the end.
The randomized algorithm can be captured by a sequence of probabilities $\{p_x\}$ with $x \in S \eqdef \{\varphi, \varphi-\epsilon, \cdots, 1\}$, meaning that the algorithm accepts deterministic box $x$ with probability $p_x$ for order $\pi$.
For the randomized algorithm to be $\Gamma$ order-competitive, it has to satisfy the following condition:
\[
\sum_{x \in S} p_x \cdot x + \left( 1-\sum_{x \in S} p_x \right) \ge \Gamma \cdot \varphi~.
\]

Next, for each $x \in S$, consider the arrival order $\pi_x$ where the deterministic boxes arrive in a decreasing order with respect to their values, and the free reward box arrive in between the $x$ box and $x-\varepsilon$ box.
It is straightforward to verify that the optimal order-aware algorithm achieves an expected reward of $x+1$ (when $\delta \to 0$) by waiting for the free reward box and accepting it if its value $1/\epsilon$ is realized. 

On the other hand, an order-unaware randomized algorithm cannot distinguish $\pi_x$ from $\pi$ until it sees the free reward box (i.e., it rejects all deterministic boxes whose values are at least $x$). 
For the randomized algorithm to be $\Gamma$ order-competitive, it has to satisfy the following condition:
\[
\sum_{y \in S: y \ge x} p_y \cdot y + \left(1- \sum_{y \in S: y \ge x} p_y \right) \cdot (x+1) \ge \Gamma \cdot (x+1)~, \quad \forall x \in S.
\]

Therefore, the order-competitive ratio of the algorithm is upper bounded by the following program:
\begin{align*}
    \max_{\{p_x\}} : \quad & \Gamma \\
    \text{subject to} : \quad & \sum_{x \in S} p_x \cdot x + \left( 1-\sum_{x \in S} p_x \right) \ge \Gamma \cdot \varphi \\
    & \sum_{y \in S: y \ge x} p_y \cdot y + \left(1- \sum_{y \in S: y \ge x} p_y \right) \cdot (x+1) \ge \Gamma \cdot (x+1)~, \quad \forall x \in S.
\end{align*}
By duality, we consider the corresponding dual program:
\begin{align*}
    \min_{\mu,\{\lambda_x\}} : \quad & \mu + \sum_{x \in S} \lambda_x \cdot (x+1) \\
    \text{subject to} : \quad & \varphi \cdot \mu + \sum_{x \in S} \lambda_x \cdot (x+1) \ge 1 \\
    & \mu \cdot (x-1) + \sum_{y \in S: y \le x} \lambda_y \cdot (x-y-1) \le 0, \quad \forall x \in S 
\end{align*}

Furthermore, we consider the limit case when $\epsilon \to 0$ (informally, the instance consists a continuous sequence of deterministic boxes), the above program then becomes continuous:
\begin{align*}
    \min_{\mu,\{\lambda(x)\}} : \quad & \mu + \int_1^\varphi \lambda(x)\cdot (x+1)\dd x\\
    \text{subject to} : \quad & \varphi \cdot \mu + \int_1^{\varphi} \lambda(x) \cdot (x+1) \dd x \ge 1 \\
    & \mu \cdot (x-1) + \int_{1}^{x} \lambda(y) \cdot (x-y-1) \dd y \le 0~, \quad \forall x \in [1, \varphi]
\end{align*}
It is straightforward to verify the feasibility of the following solution: 
\[ 
\mu = \frac{(\sqrt{5}-1)e}{3 e - \sqrt{5} e + 2
   e^{\frac{1}{2}+\frac{\sqrt{5}}{2}}}
\quad \text{and} \quad
\lambda(x) = \frac{(\sqrt{5}-1) e^x}{3 e - \sqrt{5} e + 2
   e^{\frac{1}{2}+\frac{\sqrt{5}}{2}}}~,
\]
that we omit the tedious calculation. Therefore, we establish a $\mu+\int_1^\varphi \lambda(x)\cdot(x+1)\dd x \approx 0.8293$ upper bound of the order-competitive ratio.

%% file: hardness_tvd.tex
Consider an instance with the following boxes:
\begin{itemize}
    \item A set of deterministic boxes with values $c, c-\epsilon, c-2\epsilon, \cdots, 0$, where $c \in [0.5,1]$ is a constant to be optimized later. For every $x \in \{c, c-\epsilon, c-2\epsilon, \cdots, 0\}$, we have $\frac{1-c}{\epsilon}$ identical deterministic boxes with the same value $x$.
    \item A set of identical free reward boxes, whose value equals $\frac{\epsilon}{\delta}$ with probability $\delta$, and $0$ otherwise. In total, we have $\frac{1-c}{\epsilon}$ number of such boxes.
\end{itemize}
We would be interested in the case when $\epsilon \gg \delta$ and both parameters go to $0$.
For every $x$, consider the following arrival order $\pi_x$ that consists of three stages:
\begin{itemize}
    \item First, deterministic boxes with value greater than $x$ arrive one by one in descending order of their values. That is, all $\frac{1-c}{\epsilon}$ deterministic boxes with value $c$ arrive, then all deterministic boxes with value $c - \epsilon$ arrive, and so on.
    \item Second, free reward boxes and deterministic boxes with value $x$ arrive alternatively. 
    \item Finally, remaining deterministic boxes arrive in descending order of their values.
\end{itemize}
Refer to Figure~\ref{fig:simp-Ox} for an illustration of the arrival order. We remark that the third stage is irrelevant.

\begin{figure}[H]
    \centering
    \begin{tikzpicture}

\def\boxwidth{2}
\def\boxspace{0.6}
\def\boxheight{1.2}

\draw  [fill=gray] (0, 0) rectangle ++ (\boxwidth, \boxheight) node[midway] {$c$};

\draw  [fill=gray] (\boxwidth + \boxspace, 0) rectangle ++ (\boxwidth, \boxheight) node[midway] {$c - \epsilon$};

\draw  [fill=black] (\boxwidth*2 + \boxspace + \boxwidth * 0.25, \boxheight*0.5) circle (0.03);

\draw  [fill=black] (\boxwidth*2 + \boxspace + \boxwidth * 0.5, \boxheight*0.5) circle (0.03);

\draw  [fill=black] (\boxwidth*2 + \boxspace + \boxwidth * 0.75, \boxheight*0.5) circle (0.03);

\draw  [fill=gray] (\boxwidth*3 + \boxspace, 0) rectangle ++ (\boxwidth, \boxheight) node[midway] {$x + \epsilon$};

\draw  [fill=white] (\boxwidth*4 + \boxspace*2, 0) rectangle ++ (\boxwidth, \boxheight) node[midway, align=center] {free reward \\ $\epsilon$};

\draw  [fill=gray] (\boxwidth*5 + \boxspace*3, 0) rectangle ++ (\boxwidth, \boxheight) node[midway] {$x$};

\draw[decorate,decoration={brace,amplitude=5pt,mirror,raise=1ex}] (\boxwidth*4 + \boxspace*2,0) -- (\boxwidth*6 + \boxspace*3,0) node[midway,yshift=-2em]{repeat for $\frac{1-c}{\epsilon}$ times};

\end{tikzpicture}
    \caption{Illustration of arrival order $\pi_x$}
    \label{fig:simp-Ox}
\end{figure}
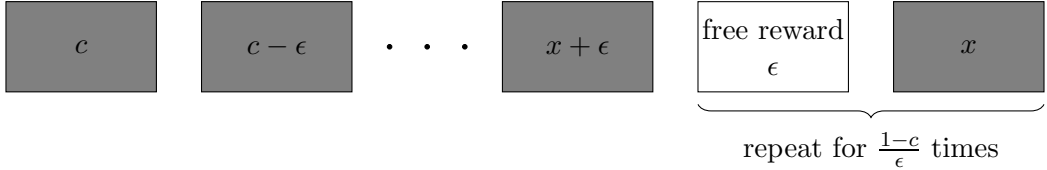

We first calculate the optimal payoff of an order-aware algorithm.
\begin{lemma}
\label{lem:tvd_opt}
    When $\delta \to 0$, $\opt(\pi_x) = 1 - c + x$.
\end{lemma}
\begin{proof}
    It is straightforward to check that the best strategy is to wait until the last deterministic box with value $x$ and accept any realized free reward beforehand. The expected reward is then
    \begin{equation*}
        (1 - \delta)^{(1-c)/\epsilon} \cdot x + \sum_{i=0}^{(1-c)/\epsilon-1} (1-\delta)^i \epsilon \overset{\delta \to 0}{\longrightarrow} 1 - c + x.
    \end{equation*}
\end{proof}

Next, we study the performance of our \emph{targeted value algorithms with detection} ($\tvd$) algorithm for arrival order $\pi_x$.

\begin{lemma}\label{lem:TVD-on-Ox}
    For arrival order $\pi_x$, when $\delta \to 0$ and $\epsilon \to 0$,
    \begin{itemize}
        \item if $c \leq g_0 \leq 1 - c + x$, then $\tvd(g_0) = g_0$;
        \item if $1 - c + x < g_0 \leq 1$, then $\tvd(g_0) = \max(1 - c, g_0 - (1-c))$.
    \end{itemize}
\end{lemma}

Before proving Lemma \ref{lem:TVD-on-Ox}, we state a few useful observations on how the targeted value would be updated throughout the whole procedure.

\begin{claim}\label{cl:TVD-hard-d}
    If the $t$-th box has deterministic value $y$, then 
    \begin{equation*}
        g_{t} = \begin{cases}
            0, & \text{ if } y \geq g_{t-1}, \\
            g_{t-1}, & \text{ if } y < g_{t-1}.
        \end{cases}
    \end{equation*}
\end{claim}

\begin{proof}
    By definition, 
    \begin{equation*}
        g_t = \min \left\{ x \ge 0 \left| \Exlong[v_t]{\max (v_t, x)} \ge g_{t-1}  \right. \right\}.
    \end{equation*}
    If $y \geq g_{t-1}$, then $\max (v_t, x) \geq g_{t-1}$ for any $x \geq 0$. If $y < g_{t-1}$, then for any $x < g_{t-1}$ we have $\max (v_t, x) < g_{t-1}$.
\end{proof}

\begin{claim}\label{cl:TVD-hard-f}
    If the $t$-th box is a free reward box, when $\delta \to 0$, we have $g_{t} = \max \Set{g_{t-1} - \epsilon, 0}$.
\end{claim}

\begin{proof}
    Recall $g_t = \min \left\{ x \ge 0 \left| \Exlong[v_t]{\max (v_t, x)} \ge g_{t-1}  \right. \right\}$.
    Since $\frac{\epsilon}{\delta} \gg g_{t-1}$, 
    $$
    \Exlong[v_t]{\max (v_t, x)} = (1 - \delta) \cdot x + \delta \cdot \frac{\epsilon}{\delta} = (1 - \delta) \cdot x + \epsilon.
    $$
    Therefore, $g_t = \max \Set{\frac{g_{t-1} - \epsilon}{1 - \delta}, 0}$. When $\delta \to 0$, $g_t = \max \Set{g_{t-1} - \epsilon, 0}$.
\end{proof}

Now, we are ready to prove Lemma \ref{lem:TVD-on-Ox}.
\begin{proofof}{Lemma \ref{lem:TVD-on-Ox}}
Consider the following three cases depending on the value of $g_0$.
\begin{enumerate}
    \item If $g_0 = c$, the algorithm accepts the first deterministic box with value $c$. Therefore $\tvd(g_0) = c$.
    \item If $c < g_0 \leq 1 - c + x$. Our algorithm never switch to the \emph{conservative} mode since $g_0 \le \opt = 1-c+x$.
    
    According to Claim \ref{cl:TVD-hard-d}, the target value $g$ would remain unchanged after observing all deterministic values in $(x, c]$.
    Then by Claim \ref{cl:TVD-hard-f}, the target value $g$ would decrease by $\epsilon$ after observing each free reward, until $g$ becomes no greater than $x$ (which means the algorithm would take the next box with deterministic value $x$). In retrospect, the algorithm would take the $\lceil \frac{g_0 - x}{\epsilon} \rceil$-th deterministic $x$, and all free reward before it. The expected reward is then $\lceil \frac{g_0 - x}{\epsilon} \rceil \cdot \epsilon + x \to g_0$ as $\epsilon \to 0$.
    \item If $1 - c + x < g_0 \leq 1$. The target value $g$ would be unchanged (i.e. the algorithm would not take any box) until the algorithm switch to the conservative mode. The switch happens exactly at the arrival of the last deterministic value $g_0 - (1-c)$. Thus, it suffices to study the single-threshold policy with threshold $\tau$ that maximizes
    \begin{equation*}
        f(\tau) \eqdef \left( \Prx{\max_{i \ge t} v_i \ge \tau} \cdot \tau + \Prx{\max_{i \ge t} v_i < \tau} \cdot \Ex{\left(\max_{i \ge t} v_i -\tau\right)^+} \right).
    \end{equation*}
    Here $\max_{i \ge t} v_i$ is the maximum value among the deterministic box $g_0 - (1-c)$, and $\frac{1 - c}{\epsilon}$ independent free rewards with expected value $\epsilon$.

    We calculate $f(\tau)$ for any $\tau \geq 0$:
    \begin{itemize}
        \item If $\tau < g_0 - (1-c)$, $\max_{i \ge t} v_i \geq g_0 - (1-c) > \tau$. Hence, $\Prx{\max_{i \ge t} v_i \ge \tau} = 1$ and $f(\tau) = \tau$.
        \item If $g_0 - (1-c) \leq \tau < \epsilon / \delta$, $\Prx{\max_{i \ge t} v_i \ge \tau} = 1 - (1 - \delta)^{(1-c)/\epsilon} = \delta \cdot \frac{1-c}{\epsilon} + o(\delta)$ when $\epsilon \gg \delta$ and $\delta \to 0$. Therefore,
        \begin{align*}
            f(\tau) = & \frac{(1-c)\delta}{\epsilon} \cdot \tau + \left(1 - \frac{(1-c)\delta}{\epsilon} \right) \left( \frac{1-c}{\epsilon} \cdot \epsilon - \frac{(1-c)\delta}{\epsilon} \cdot \tau \right) + o(\delta) \\
            = & 1 - c + O(\delta) \to 1-c.
        \end{align*}
        \item If $\tau \geq \epsilon / \delta$, such $\tau$ is too large so that $\Prx{\max_{i \ge t} v_i \ge \tau} = 0$ and $f(\tau) = 0$.
    \end{itemize}

    Therefore, $\max_{\tau} f(\tau) = \max(g_0 - (1-c), 1 - c)$, and $\tau$ would either be $g_0 - (1 - c)$ (minus an infinitesimal value), or at least $g_0 - (1 - c)$. One can verify that in the former case $\tvd(g_0) = g_0 - (1 - c)$, and in the latter case $\tvd(g_0) = 1 - c$.
\end{enumerate}
\end{proofof}

Finally, we establish an upper bound of the order competitive ratio of $\tvd$ by considering all possible arrival orders $\pi_x$. According to Lemma~\ref{lem:tvd_opt} and ~\ref{lem:TVD-on-Ox}. Its order competitive ratio is upper bounded by the following, where $\rho$ corresponds to the probability density function of $g_0$.
\begin{align*}
    \max_{\{\rho(x)\}} : \ & \Gamma \\
    \text{subject to} :\ & \int_{c}^{1-c+x} y\rho(y)\dd y+\int_{1-c+x}^1 \max(1-c, y-(1-c))\rho(y)\dd y \\
    & \qquad \ge \Gamma \cdot (1 - c + x)~, & \forall x \in (2c-1, c] \\
    & \int_c^{1} \rho(x) \dd x \le 1.
\end{align*}

We establish an upper bound of the program by considering its dual program:
\begin{align*}
    \min_{\mu, \{ \lambda(x) \}} : \ & \mu \\
    \text{subject to} : \ &  y \int_{y-(1-c)}^{c} \lambda(x) \dd x + \max (1-c, y - (1-c))\int_{2c-1}^{y-(1-c)} \lambda(x) \dd x \leq \mu, \quad \forall y \in [c, 1] \\
    \ & \int_{2c-1}^c (1 - c + x) \lambda(x) \dd x \geq 1.
\end{align*}

Let $c \approx 0.583027$ be the solution of $c = -1 + \frac{1}{2(1-c)} + (1-c) \ln \left( \frac{1-c}{2c-1} \right)$ and $\lambda(x)$ be that
\begin{equation*}
    \lambda(x) = \begin{cases}
        \frac{a}{x^2}, & 2c-1 \leq x < 1-c, \\
        \frac{b}{1-c}, & 1-c \leq x \leq c,
    \end{cases}
\end{equation*}
where $(a, b) \approx (0.215941, 1.300426)$ is the solution to the linear system
\begin{equation*}
    \begin{cases}
        & a \int_{2c-1}^{1-c} x^{-2} \dd x + b \int_{1-c}^c \frac{1}{1-c} \dd x = b,\\
        & a\int_{2c-1}^{1-c} \frac{1-c+x}{x^2} \dd x + b\int_{1-c}^{c} \frac{1-c+x}{1-c} \dd x =1.
    \end{cases}
\end{equation*}
It is straightforward to verify the feasibility of the solution that we omit the tedious calculations. This results in an upper bound of $\mu = c \int_{2c-1}^c \lambda(x) \dd x \approx 0.758184$ on the order-competitive ratio of $\tvd$.